%% file: ieepcsg.tex
\documentclass[pre,onecolumn,showpacs,superscriptaddress,preprintnumbers,floatfix]{revtex4}

\usepackage{ifpdf}
\usepackage{hyperref}
\usepackage{dcolumn}
\usepackage{url}
\usepackage{amsmath}
\usepackage{amscd}
\usepackage{amsfonts}
\usepackage{amssymb}
\usepackage{bm}   
\usepackage{bbm}   
\usepackage{verbatim}
\usepackage{stmaryrd}
\usepackage{amsthm}
\usepackage{color}
\usepackage{footnote}

\ifpdf
  \usepackage[pdftex]{graphicx}         
\else
  \usepackage[dvips]{graphicx}
\fi

\input{cmechabbrev}

\input{shortcuts}

\begin{document}

\title{Infinite Excess Entropy Processes with Countable-State Generators}

\author{Nicholas F. Travers}
\email{ntravers@math.ucdavis.edu}
\affiliation{Complexity Sciences Center}
\affiliation{Mathematics Department}

\author{James P. Crutchfield}
\email{chaos@ucdavis.edu}
\affiliation{Complexity Sciences Center}
\affiliation{Mathematics Department}
\affiliation{Physics Department\\
University of California at Davis,\\
One Shields Avenue, Davis, CA 95616}
\affiliation{Santa Fe Institute\\
1399 Hyde Park Road, Santa Fe, NM 87501}

\date{\today}

\bibliographystyle{unsrt}

\begin{abstract}

We present two examples of finite-alphabet, infinite excess entropy processes
generated by invariant hidden Markov models (HMMs) with countable state sets.
The first, simpler example is not ergodic, but the second is. It appears these
are the first constructions of processes of this type. Previous examples of
infinite excess entropy processes over finite alphabets admit only invariant
HMM presentations with uncountable state sets. 

\vspace{0.1in}


\end{abstract}

\preprint{Santa Fe Institute Working Paper 11-11-XXX}
\preprint{arxiv.org:1111.XXXX [XXXX]}
\maketitle


\vspace{- 5 mm}
\section{Introduction}
\label{sec:Introduction}

For a stationary process $(\MS_t)$ the \emph{excess entropy} $\EE$ is the
mutual information between the infinite past $\Past = \ldots \MS_{-2} \MS_{-1}$
and the infinite future $\Future = \MS_0 \MS_1 \ldots$ . It has a long history
and is widely employed as a measure of correlation and complexity in a variety
of fields, from ergodic theory and dynamical systems to neuroscience and
linguistics \cite{Junc79,Crut82b,Gras86,Lind88,Bial01a,Debo11a}; see Ref.
\cite{Crut01a} and references therein for a review. 

An important question in classifying a given process is whether it is
\emph{finitary} (finite excess entropy) or \emph{infinitary} (infinite excess
entropy). Over a finite alphabet, many of the simple process classes commonly
studied are always finitary. These include all i.i.d. processes, Markov chains,
and processes with finite-state hidden Markov model (HMM) presentations. There
also exist several well known examples of finite-alphabet infinitary processes,
though. For instance, the symbolic dynamics at the onset of chaos in the
logistic map and similar dynamical systems \cite{Crut01a} and the stationary
representation of the binary Fibonacci sequence \cite{Wern97} are both
infinitary. 

These latter processes, however, only admit invariant HMM presentations
\footnotemark[1] with uncountable state sets. Indeed, any process
generated by an invariant countable-state HMM either has positive entropy rate
or consists entirely of periodic sequences, which these do not; see App.
\ref{AppendixB}. Versions of the \emph{Santa Fe Process} introduced in Ref.
\cite{Debo11a} are finite-alphabet infinitary processes with positive entropy
rate. However, they were not constructed directly as HMMs, and it seems
unlikely that they should have any invariant countable-state presentations.
To the best of our knowledge, to date there are no examples of finite-alphabet,
infinitary processes with invariant countable-state presentations. 

\footnotetext[1]{\emph{Invariant} here means that the state sequence of the
underlying Markov chain and, hence, the output sequence generated by the HMM
are both stationary processes \cite{Lohr10a}.}

We present two such examples. The first is nonergodic, and the information
conveyed from the past to the future essentially consists of the ergodic
component along a given realization. This example is straightforward to
construct and, though previously unpublished, we suspect that others are
aware of this or similar constructions. The second, ergodic example,
though, is more involved and we believe that both its structure and properties
are novel. 

To put these contributions in perspective, note that \emph{any} stationary
finite-alphabet process may be trivially represented as an invariant HMM with
an uncountable state set, in which each infinite history $\past$ corresponds to
a single state. Thus, it is clear invariant HMMs with uncountable state sets
can generate finite-alphabet infinitary processes. In contrast, for any
finite-state HMM $\EE$ is always finite---bounded by the logarithm of the
number of states. The case of countable-state HMMs lies in between the
finite-state and uncountable-state cases, and it was previously not clear
whether it is possible to have countable-state invariant HMMs that generate
infinitary finite-alphabet processes and, in particular, ergodic ones. Here,
we show that infinite excess entropy is indeed possible for processes with
countable-state state generators, even when ergodicity is required. 

\section{Background}
\label{sec:Background}

\subsection{Excess Entropy} 
\label{subsec:ExcessEntropy}

\begin{Def}
For a stationary, finite-alphabet process $(\MS_t)_{t \in \Z}$ the excess entropy $\EE$ is the mutual information between the infinite past
$\Past = ... \MS_{-2} \MS_{-1}$ and the infinite future
$\Future = {\MS_0} {\MS_1} ...$:
\begin{align}
\label{eq:Emut}
\EE = I[\Past;\Future] = \lim_{t \to \infty} I[\Past^t;\Future^t] ~,
\end{align} 
where $\Past^t = \MS_{-t} ... \MS_{-1}$ and $\Future^t = \MS_0 ...
\MS_{t-1}$ are the length-$t$ past and future, respectively.  
\end{Def}
In Refs. \cite{Crut01a,Lohr10a} it is shown that $\EE$ may 
also be expressed alternatively as: 
\begin{align}
\label{eq:Elim}
\EE = \lim_{t \to \infty} \left( H[\Future^t] - \hmu t \right) ~,
\end{align}
where $\hmu$ is the process \emph{entropy rate}:
\begin{align}
\label{eq:hmu}
\hmu = \lim_{t \to \infty} \frac{H[\Future^t]}{t} = \lim_{t \to \infty} H[\MS_t|\Future^t] ~.
\end{align}
That is, the excess entropy $\EE$ is the asymptotic amount of entropy (information)
in length-$t$ blocks of random variables beyond that explained by the entropy rate. 
The excess entropy derives its name from this formulation. We also use this
formulation to establish that the process of Sec. \ref{subsec:Example1} is
infinitary. 

Expanding the block entropy $H[\Future^t]$ in Eq. (\ref{eq:Elim}) with the
chain rule and recombining terms gives another important formulation:
\begin{align}
\EE = \sum_{t = 1}^{\infty} \left( \hmu(t) - \hmu \right) ~,
\label{eq:Esum}
\end{align}
where $\hmu(t)$ is the \emph{length-$t$ entropy-rate approximation}:
\begin{align}
\label{eq:hmut}
\hmu(t) = H[\MS_{t-1}|\Future^{t-1}] ~,
\end{align}
the conditional entropy in the $t_{th}$ symbol given the previous
$t-1$ symbols. This final formulation will be used to establish that the
process of Sec. \ref{subsec:Example2} is infinitary. 

\subsection{Hidden Markov Models}
\label{subsec:HiddenMarkovModels}

There are two primary types of hidden Markov models: edge-emitting (or
\emph{Mealy}) and state-emitting (or \emph{Moore}). We
work with the former edge-emitting type, but the two are equivalent
in that any model of one type over a finite alphabet may converted to a
model of the other type without changing the cardinality of the state set
by more than a constant factor---the alphabet size. Thus, for our purposes, Mealy
HMMs are sufficiently general. We also consider only \emph{invariant HMMs},
as defined in \cite{Lohr10a}, over finite alphabets and with countable state
sets. 

\begin{Def}
An \emph{invariant, edge-emitting, countable-state, finite-alphabet hidden
Markov model} (hereafter referred to simply as a \emph{countable-state HMM})
is a 4-tuple $(\CSSet, \XX, \{T^{x}\}, \pi)$ where:
\begin{enumerate}
\item $\CSSet$ is a countable set of states.
\item $\XX$ is a finite alphabet of output symbols.
\item $T^{(x)}, x \in \XX$, are symbol labeled transition matrices.
	$T^{(x)}_{\cs \cs'}$ is the probability that state $\cs$ transitions
	to state $\cs'$ on symbol $x$.
\item $\pi$ is an \emph{invariant} or \emph{stationary} distribution for the underlying Markov chain 
          over states with transition matrix $T = \sum_{x \in \XX} T^{(x)}$.
That is, $\pi$ satisfies $\pi = \pi T$. 
\end{enumerate}
\end{Def} 

\begin{Rem}
``Countable'' in Property 1 means either finite or countably infinite. If the
state set $\CSSet$ is finite, we also refer to the HMM as \emph{finite-state}.
\end{Rem}

A hidden Markov model may be depicted as a directed graph with labeled
edges. The vertices are the states $\causalstate \in \CSSet$ and, for all 
$\causalstate, \causalstate' \in \CSSet$ with
$T^{(x)}_{\causalstate \causalstate'} > 0$, there is a directed edge from
state $\causalstate$ to state $\causalstate'$ labeled $p|x$ for the symbol
$x$ and transition probability $p = T^{(x)}_{\causalstate \causalstate'}$. 
These probabilities are normalized so that the sum of probabilities on all
outgoing edges from each state is $1$. 
An example is given in Fig. \ref{fig:EvenProcess}. 

\begin{figure}[h]
\begin{minipage}[b]{0.4\textwidth}
\centering
\includegraphics[scale=0.75]{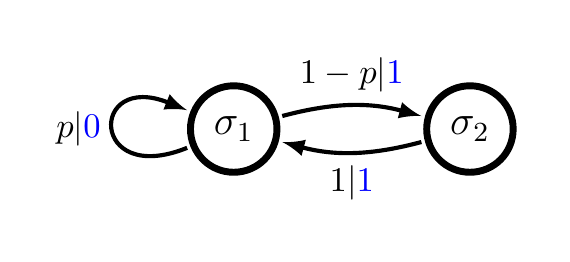}
\end{minipage}
\hspace{0.5cm}
\begin{minipage}[b]{0.4\textwidth}
\centering
\begin{align*}
T^{(0)} & =
  \left(
  \begin{array}{cc}
  	p & 0 \\
	0 & 0 \\
  \end{array}
  \right) \nonumber \\
T^{(1)} & =
  \left(
  \begin{array}{cc}
	0 & 1-p \\
	1 & 0 \\
  \end{array}
  \right) 
\end{align*}
\end{minipage}
\vspace{2 mm}
\caption{A hidden Markov model (the \eM) for the Even Process. The support for
  this process consists of all binary sequences in which blocks of
  uninterrupted $1$s are even in length, bounded by $0$s.  After each 
  even length is reached, there is a probability $p$ of breaking the block
  of $1$s by inserting a $0$. The machine has two internal states
  $\CSSet = \{ \cs_1, \cs_2 \}$, a two symbol alphabet $\XX = \{0,1\}$, and
  a single parameter $p \in (0,1)$ that controls the transition probabilities.
  The associated Markov chain over states is finite-state and irreducible and,
  thus, has a unique stationary distribution
  $\pi = (\pi_1, \pi_2) = \left( 1/(2-p), (1-p)/(2-p) \right)$.
  The graphical representation of the machine is given on the left, with the
  corresponding transition matrices on the right. In the graphical
  representation the \textcolor{blue}{symbols} labeling the transitions have
  been colored \textcolor{blue}{blue}, for visual contrast, while the
  transition probabilities are black.
  } 
\label{fig:EvenProcess}
\end{figure}

The operation of a HMM may be thought of as a weighted random walk on the
associated graph. That is, from the current state $\cs$ the next state $\cs'$
is determined by following an outgoing edge from $\cs$ according to the edges'
relative probabilities (or weights). During the transition, the HMM
outputs the symbol $x$ labeling this edge. 

The state sequence $(\CS_t)$ determined in such a fashion is simply a Markov chain with transition 
matrix $T$. However, we are interested not simply in the state sequence of the HMM, but rather the 
associated sequence of output symbols $(\MS_t)$ that are generated by reading the labels off the 
edges as they are followed. The interpretation is that an observer of the HMM may directly observe
this sequence of output symbols, but not the hidden internal states. Alternatively, one may consider the
Markov chain over edges $(E_t)$, of which the observed symbol sequence $(\MS_t)$ is simply a projection. 

In either case, the process $(\MS_t)$ generated by the HMM
$(\CSSet, \XX, \{T^{x}\}, \pi)$ is defined as the output sequence of edge
symbols, which results from running the Markov chain over states according to
the stationary law with marginals $\P(\CS_0) = \P(\CS_t) = \pi$. 
It is easy to verify that this process is itself stationary, with word probabilities given by: 
\begin{align}
\label{eq:Pw}
\P(w) = \norm{\pi T^{(w)}}_1 ~,
\end{align}
where for a given word $w = w_1 ... w_n \in \XX^+$, $T^{(w)}$ is the
word transition matrix $T^{(w)} = T^{(w_1)} \cdot \cdot \cdot T^{(w_n)}$. 
The \emph{process language} is the set of words $\L = \{ w : \P(w) > 0 \}$.

\begin{Rem}
Even for a \emph{noninvariant HMM} $(\CSSet, \XX, \{T^{x}\}, \pi)$, where the
state distribution $\pi$ is not stationary, one may always define a
one-sided process $(\MS_t)_{t \geq 0}$ with marginals given by:
\begin{align}
\P(\Future^{|w|} = w) = \norm{\pi T^{(w)}}_1 ~. 
\end{align}
Furthermore, though the state sequence $(\CS_t)_{t \geq 0}$ will not be a
stationary process if $\pi$ is not a stationary distribution for $T$, the
output sequence $(\MS_t)_{t \geq 0}$ may still be stationary. In fact,
Ref. \cite[Example 2.9]{Lohr10a} showed that any one-sided process over a
finite alphabet $\XX$, stationary or not, may be represented as a countable-state
noninvariant HMM in which the states correspond to finite-length words in $\XX^+$,
of which there are only countably many. By stationarity, a one-sided stationary process
generated by such a noninvariant HMM can be uniquely extended to a two-sided 
stationary process. So, in a sense, any two-sided stationary process $(\MS_t)_{t \in \Z}$
can be said to be generated by a noninvariant countable-state HMM. 
Though, this is a slightly unnatural interpretation of process generation in that 
the two-sided process $(\MS_t)_{t \in \Z}$ is not directly the process obtained 
by reading symbols off the edges of the HMM as it runs along transitioning between 
states in bi-infinite time. In either case, the space of stationary finite-alphabet
processes generated by noninvariant countable-state HMMs is too large: it
includes all stationary finite-alphabet processes. Due to this, we restrict
to the case of invariant HMMs where both the state sequence $(\CS_t)$ and
output sequence $(\MS_t)$ are stationary. Clearly, if one allows
finite-alphabet processes generated by noninvariant countable-state HMMs
there are infinitary examples. And so, in the following development HMM will
implicitly mean invariant HMM, but this will no longer be stated. 
\end{Rem}

We consider now an important property known as unifilarity. This property is
useful in that many quantities are analytically computable only for unifilar
HMMs. In particular, for unifilar HMMs the entropy rate $\hmu$ is often
directly computable, unlike the nonunifilar case. Both of the examples
constructed in Sec. \ref{sec:Examples} are unifilar, as is the Even Process
HMM of Fig. \ref{fig:EvenProcess}. 

\begin{Def}
A HMM $(\CSSet, \XX, \{T^{x}\}, \pi)$ is \emph{unifilar} if for each
$\cs \in \CSSet$ and $x \in \XX$
there is at most one outgoing edge from state $\cs$ labeled with symbol $x$ in the associated graph $G$.  
\end{Def}
It is well known that for any finite-state unifilar HMM the entropy rate in
the output process $(\MS_t)$ is simply the conditional entropy in the next
symbol given the current state:
\begin{align}
\label{eq:UnifilarEntropyRateFormula}
\hmu = H[\MS_0|\CS_0] = \sum_{\cs \in \CSSet} \pi_{\cs} h_{\cs} ~,
\end{align}
where $\pi_{\cs}$ is the stationary probability of state $\cs$ and
$h_{\cs} = H[\MS_0|\CS_0 = \cs]$ is the conditional entropy in the next
symbol given that the current state is $\cs$. 

We are unaware, though, of any proof that this is generally true for
countable-state HMMs. If the entropy in the stationary distribution $H[\pi]$
is finite, then a proof along the lines given in Ref. \cite{Trav11b} 
carries through to the countable-state case and Eq.
(\ref{eq:UnifilarEntropyRateFormula}) still holds. 
However, countable-state HMMs may sometimes have $H[\pi] = \infty$.
Furthermore, it can be shown \cite{Lohr10a} that the excess entropy $\EE$
is always bounded above by $H[\pi]$. So, for the infinitary process
of Sec. \ref{subsec:Example2} we need slightly more than unifilarity to establish the value of $\hmu$. 
To this end, we consider a property known as \emph{exactness} \cite{Trav11a}. 

\begin{Def}
A HMM is said to be \emph{exact} if for a.e. infinite future
$\future = \ms_0 \ms_1 ...$ generated by the HMM an observer synchronizes
to the internal state after a finite time. That is, for a.e.  $\future$ there
exists $t \in \N$ such that $H[\CS_t|\Future^t = \future^t] = 0$, where 
$\future^t = \ms_0 \ms_1 ... \ms_{t-1}$ denotes the the first $t$ symbols of a given $\future$.  
\end{Def}

In App. \ref{AppendixA} we prove the following proposition. 

\begin{Prop}
For any countable-state, exact, unifilar HMM, the entropy rate is given by
the standard formula of Eq. (\ref{eq:UnifilarEntropyRateFormula}). 
\label{prop:ExactEntropyRateFormula}
\end{Prop}

The HMM constructed in Sec. \ref{subsec:Example2} is both exact and unifilar, 
so Prop. \ref{prop:ExactEntropyRateFormula} applies. Using this explicit formula 
for $\hmu$, we will show that $\EE = \sum_{t=1}^{\infty} \left( \hmu(t) - \hmu \right)$ is infinite.

\section{Constructions}
\label{sec:Examples}

We present two constructions of (invariant) countable-state HMMs that generate
infinitary processes. In the first example the output process is not ergodic,
but in the second it is. 

\subsection{Heavy-Tailed Periodic Mixture: An infinitary nonergodic process with a countable-state presentation}
\label{subsec:Example1}

Figure \ref{fig:NonErgodic} depicts a countable-state HMM $M$, for a nonergodic
infinitary process $\Process$. The machine $M$ consists of a countable collection of disjoint 
strongly connected subcomponents $M_i$, $i \geq 2$. For each $i$, the component $M_i$ 
generates the periodic processes $\Process_i$ consisting of $i-1$ 1s followed by a $0$. 
The weighting over components is taken as a heavy-tailed distribution with infinite entropy. 
For this reason, we refer to the process $M$ generates as the 
\emph{Heavy-Tailed Periodic Mixture} (HPM) process.

\begin{figure}[h]
\includegraphics[scale=0.7]{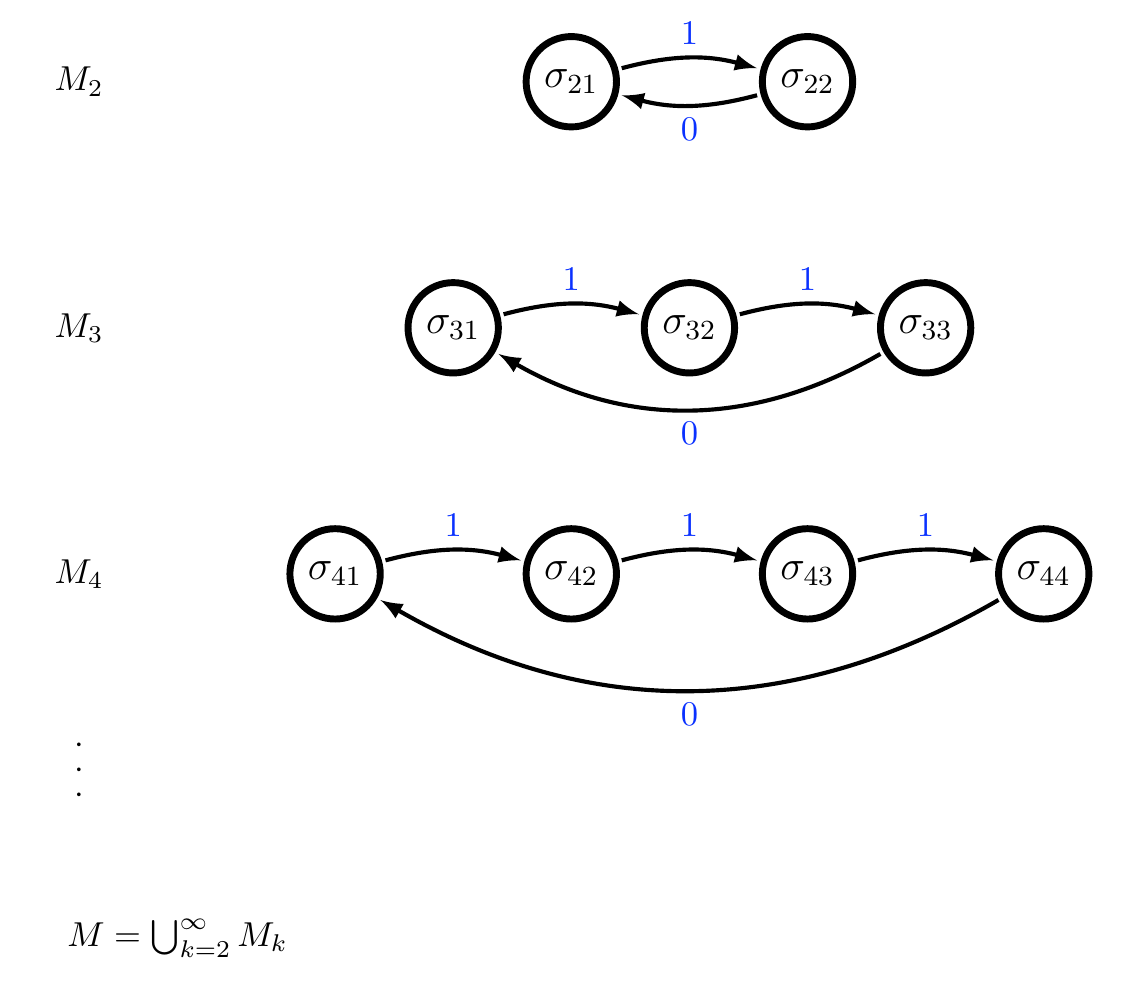}
\caption{A countable-state HMM for the Heavy-Tailed Periodic Mixture Process. 
  The machine $M$ is the union of the machines $M_i, i \geq 2$, generating the
  period-$i$ processes of $i-1$ $1$s followed by a $0$. All topologically
  allowed transitions have probability $1$. So, for visual clarity these
  probabilities are omitted from the edge labels and only the symbols labeling
  the transitions are given. The stationary distribution $\pi$ is chosen such
  that the combined probability $\mu_i$ of all states in the the $i_{th}$
  component is $\mu_{i} = C/(i \log^2 i)$, where
  $C = 1/ \left( \sum_{i=2}^{\infty} 1/(i \log^2 i) \right)$ is a normalizing
  constant. Formally, the HMM $M = (\CSSet, \XX, \{T^{(x)}\}, \pi)$ has alphabet
  $\XX = \{0,1\}$, state set
  $\CSSet = \{\cs_{ij}: i \geq 2, 1 \leq j \leq i \}$,
  stationary distribution $\pi$ defined by $\pi_{ij} = C/(i^2 \log^2 i)$, and
  transition probabilities $T^{(1)}_{ij,i(j+1)} = 1$ for $i \geq 2$ and
  $1 \leq j < i$,  $T^{(0)}_{ii,i1} = 1$ for $i \geq 2$, and all other
  transitions probabilities $0$. Note that all logs here (and throughout) are
  taken base $2$, as is typical when using information-theoretic quantities.}
\label{fig:NonErgodic} 
\end{figure}

Intuitively, the information transmitted from the past to the future for
the HPM Process is the ergodic component $i$ along with the phase of the
period-$i$ process $\Process_i$ in this component. This is more information
than simply the ergodic component $i$, which is itself an infinite amount of
information: $H[(\mu_2, \mu_3, ... ,)] = \infty$. Hence, $\EE$ should be
infinite. This intuition can be made precise using the ergodic decomposition
theorem of Debowski \cite{Debo09}, but we present a more direct proof here.

\begin{Prop}
The HPM Process has infinite excess entropy. 
\label{prop:NonErgodicInfiniteE}
\end{Prop}

\begin{proof}
For the HPM Process $\Process$ we will show that
(i) $\lim_{t \to \infty} H[\Future^t] = \infty$ and (ii) $\hmu = 0$.
The conclusion then follows immediately from Eq. (\ref{eq:Elim}). To this end,
we define sets:
\begin{align*}
&  W_{i,t}  = \{w: |w| = t
	\mbox{ and } w \mbox{ is in the support of process } \Process_i \}, \\
&  U_{t} = \bigcup_{2 \leq i \leq t/2} W_{i,t} ~,~\text{and}\\
&  V_t = \bigcup_{i > t/2} W_{i,t} ~.
\end{align*} 

Note that any word $w \in W_{i,t}$ with $i \leq t/2$ contains at least two 0s. Therefore:
\begin{enumerate}
\item No two distinct states $\cs_{ij}$ and $\cs_{ij'}$ with $i \leq t/2$ generate the same length $t$ word.
\item The sets $W_{i,t}, i \leq t/2$, are disjoint from both each other and $V_t$.
\end{enumerate}
It follows that each word $w \in W_{i,t}$, with  $i \leq t/2$, can only be
generated from a single state $\cs_{ij}$ of the HMM and has probability:
\begin{align}
\P(w) & = \P(\Future^t = w) \nonumber \\
  & = \P(\CS_0 = \cs_{ij}) \cdot \P(\Future^t = w |\CS_0 = \cs_{ij}) \nonumber \\
  & = \pi_{ij} \cdot 1 \nonumber \\
  & = C/(i^2 \log^2 i) ~.
\end{align} 
Hence, for any fixed $t$:
\begin{align*}
H[\Future^t] & = \sum_{|w| = t} \P(w) \log \left( \frac{1}{\P(w)} \right) \\
  & \geq \sum_{i = 2}^{\lfloor t/2 \rfloor} \sum_{w \in W_{i,t}}
  \frac{C}{i^2 \log^2(i)} \log \left( \frac{i^2 \log^2(i)}{C}
  \right) \\
  & = \sum_{i = 2}^{\lfloor t/2 \rfloor} \frac{C}{i \log^2(i)} \log \left( \frac{i^2 \log^2(i)}{C} \right) ~,
\end{align*}
so:
\begin{align}
\lim_{t \to \infty} H[\Future^t] \geq \sum_{i = 2}^{\infty} \frac{C}{i \log^2(i)} \log \left( \frac{i^2 \log^2(i)}{C} \right)  = \infty ~,
\end{align}
which proves Claim (i). Now, to prove Claim (ii) consider the quantity:
\begin{align}
\hmu(t+1) 
& =  H[\MS_t|\Future^t] \nonumber \\
& = \sum_{w \in U_t} \P(w) \cdot H[\MS_t|\Future^t = w] + \sum_{w \in V_t} \P(w) \cdot H[\MS_t|\Future^t = w] ~.
\label{eq:HFnPm_decomposition}
\end{align} 
On the one hand, for $w \in U_t$, $H[\MS_t|\Future^t = w] = 0$ since the
current state and, hence, entire future are completely determined by any
word $w \in U_t$. On the other hand, for $w \in V_t$, 
$H[\MS_t|\Future^t = w] \leq 1$ since the alphabet is binary.
Moreover, the combined probability of all words in the set $V_t$ is simply
the probability of starting in some component $M_i$ with $i > t/2$:
$\P(V_t) = \sum_{i > t/2} \mu_i$. Thus, by Eq. (\ref{eq:HFnPm_decomposition}),
$\hmu(t+1) \leq  \sum_{i > t/2} \mu_i$. Since $\sum_i \mu_i$ converges, 
it follows that $\hmu(t) \searrow 0$, which verifies Claim (ii). 
 
\end{proof}

\subsection{Branching Copy Process: An infinitary ergodic process with a countable-state presentation}
\label{subsec:Example2}

Figure \ref{fig:Ergodic} depicts a countable-state HMM $M$ for the ergodic,
infinitary \emph{Branching Copy Process}. Essentially, the machine $M$ consists
of a binary tree with loop backs and a self-loop on the root node. From the
root node a path is chosen down the tree with each left-right (or $0$-$1$)
choice equally likely. But, at each step there is also a chance of turning back
towards the root. The path back is a not a single step, however. It has length
equal to the number of steps taken 
down the tree before returning back, and copies the path taken down
symbol-wise with $0$s replaced by $2$s and $1$s replaced by $3$s.
There is also a high self-loop probability at the root node on symbol $4$,
so some number of 4s will normally be generated after returning to the root
node before preceding again down the tree.  The process generated by this 
machine is referred to as the Branching Copy (BC) Process, because
the branch taken down the tree is copied on the loop back to the root. 

\begin{figure}[h]
\includegraphics[scale=0.76]{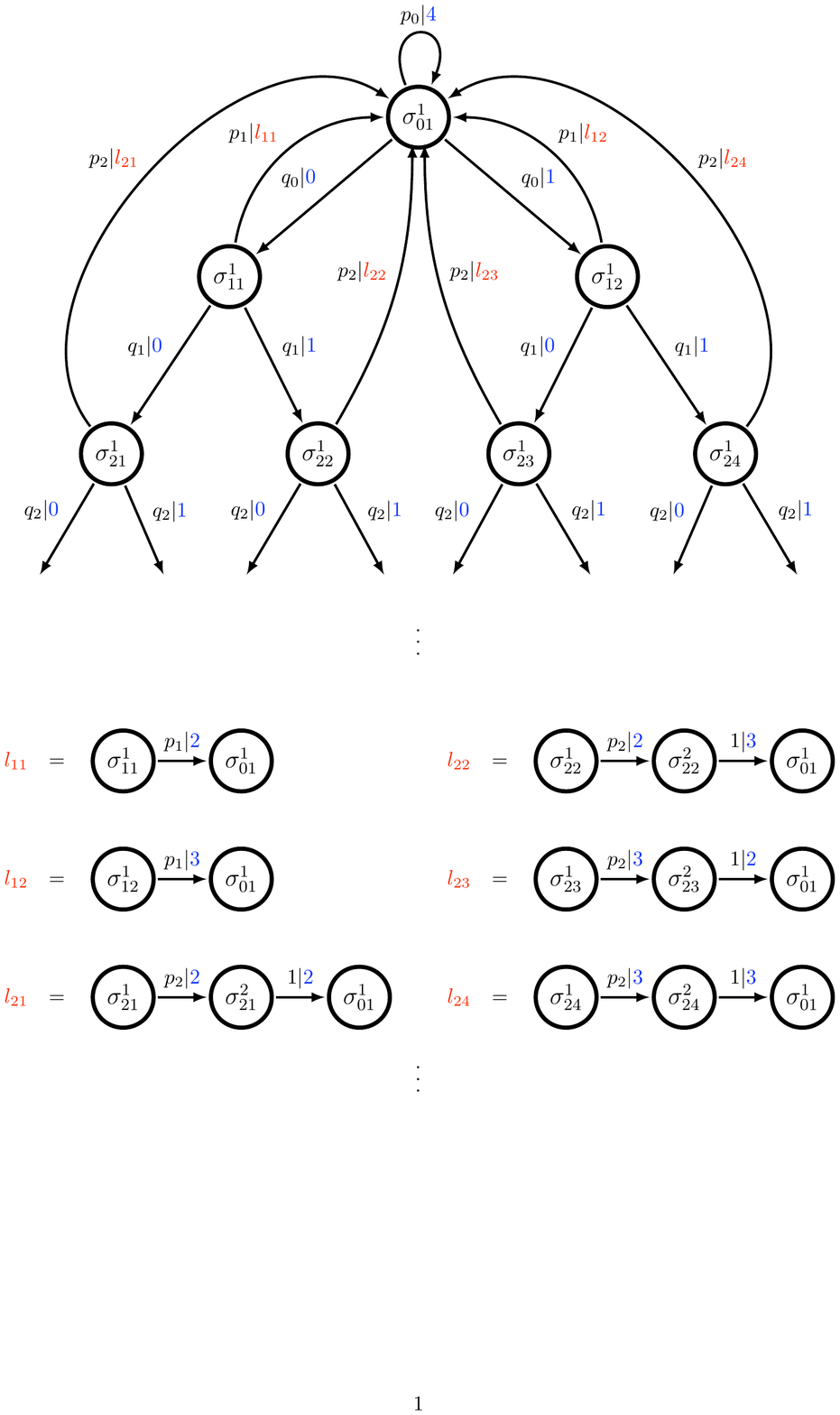}
\caption{A countable-state HMM for the Branching Copy Process. 
  The machine $M$ is essentially a binary tree with loop-back
  paths from each node in the tree to the root node and a self-loop on the
  root. At each node $\cs_{ij}^1$ in the tree there is a probability $2 q_i$
  of continuing down the tree and a probability $p_i = 1 - 2 q_i$ of turning
  back towards the root $\cs_{01}^1$ on path
  $l_{ij} \sim \cs_{ij}^1 \rightarrow \cs_{ij}^2 \rightarrow \cs_{ij}^3 ... \rightarrow \cs_{ij}^i \rightarrow \cs_{01}^1$.
  If the choice is made to head back, the next $i-1$ transitions are deterministic. The path of 0s and 1s taken to get from $\cs_{01}^1$ to
  $\cs_{ij}^1$ is copied on the return with $0$s replaced by $2$s and $1$s
  replaced by $3$s. Formally, the alphabet is $\XX = \{0,1,2,3,4\}$ and the
  state set is
  $\CSSet = \{\cs_{ij}^k: i \geq 0, 1 \leq j \leq 2^i, 1 \leq k \leq \max\{i,1\} \}$.
  The nonzero transition probabilities are as depicted graphically with
  $p_i = 1 - 2 q_i$ for all $i \geq 0$, $q_i = i^2/[2(i+1)^2]$ for all
  $i \geq 1$, and $q_0 > 0$ taken sufficiently small so that
  $H[(p_0,q_0,q_0)] \leq 1/300$. The graph is strongly connected so the Markov
  chain over states is irreducible.
  Claim \ref{cla:PositiveRecurrent} shows that the Markov chain is also positive
  recurrent and, hence, has a unique stationary distribution $\pi$.
  Claim \ref{cla:StationaryDistribution} gives the form of $\pi$.
  }
\label{fig:Ergodic} 
\end{figure}

By inspection we see that the machine is unifilar with \emph{synchronizing word}
$w = 4$, i.e. $H[\CS_1|\MS_0 = 4] = 0$. Since the underlying Markov chain over
states $(\CS_t)$ is positive recurrent, the state sequence $(\CS_t)$ and symbol
sequence $(\MS_t)$ are both ergodic. Thus, a.e. infinite future
$\future$ contains a $4$, so the machine is exact. Therefore,
Prop. \ref{prop:ExactEntropyRateFormula} may be applied, and we know the
entropy rate $\hmu$ is given by the standard formula of Eq.
(\ref{eq:UnifilarEntropyRateFormula}): $\hmu = \sum_{\cs} \pi_{\cs} h_{\cs}$.
Since $\P(\CS_t = \cs) = \pi_{\cs}$ for any $t \in \N$, we may alternatively represent this entropy rate as:
\begin{align}
\hmu & = \sum_{\cs} \left( \sum_{w \in \L_t}
	\P(w) \phi(w)_{\cs} \right) h_{\cs} \nonumber \\
  & = \sum_{w \in \L_t} \P(w)
	\left( \sum_{\cs} \phi(w)_{\cs} h_{\cs} \right) \nonumber \\
  & = \sum_{w \in \L_t} \P(w) \th_w ~,
\label{eq:hmu_expansion}
\end{align}
where $\L_t = \{w: |w| = t, \P(w) > 0\}$ is the set of length $t$ words in
the process language $\L$, $\phi(w)$ is the conditional state distribution
induced by the word $w$ (i.e., $\phi(w)_{\cs} = \P(\CS_t = \cs|\Future^t = w$)),
and $\th_w = \sum_{\cs} \phi(w)_{\cs} h_{\cs}$ is the $\phi(w)$-weighted average 
entropy in the next symbol given knowledge of the current state $\cs$. 

Similarly, for any $t \in \N$ the entropy-rate approximation $\hmu(t+1)$ may be expressed as:
\begin{align}
\label{eq:hmut_expansion} 
\hmu(t+1) & = H[\MS_t|\Future^t] = \sum_{w \in \L_t} \P(w) h_w ~,
\end{align}
where $h_w = H[\MS_t|\Future^t = w] = H[X_0|\CS_0 \sim \phi(w)]$ is the entropy in the next symbol given the word $w$. 
Combining Eqs. (\ref{eq:hmu_expansion}) and (\ref{eq:hmut_expansion}) we have for any $t \in \N$:
\begin{align}
\label{eq:hmut_hmu_diff}
\hmu(t+1) - \hmu =  \sum_{w \in \L_t} \P(w) (h_w - \th_w) ~.
\end{align}
By concavity of the entropy function, the quantity $h_w - \th_w$ is always nonnegative. 
Furthermore, in Claim \ref{cla:hw_and_thw} we show that $h_w - \th_w$ is always bounded below 
by some fixed positive constant for any word $w$ consisting entirely of 2s and 3s. 
Also, in Claim \ref{cla:ProbWt} we show that $\P(W_t)$ scales as $1/t$, where $W_t$ is 
the set of length-$t$ words consisting entirely of 2s and 3s. Combining these results it follows that
$\hmu(t+1) - \hmu ~ \widetilde{\geq} ~1/t$ and, hence, the sum
$\EE = \sum_{t=1}^{\infty} \left( \hmu(t) - \hmu \right)$ is infinite. 

A more detailed analysis with the claims and their proofs is given below. In this we will use the following notation:
\begin{itemize}
\item  $\P_{\cs}(\cdot) = \P(\cdot|\CS_0 = \cs)$, 
\item  $V_t = \{w \in \L_t: w \mbox{ contains only 0s and 1s} \}$ and
	$W_t = \{w \in \L_t: w \mbox{ contains only 2s and 3s} \}$,
\item $\pi_{ij}^k = \P(\sigma_{ij}^k)$ is the stationary probability of
	state $\sigma_{ij}^k$,
\item $R_{ij} = \{ \sigma_{ij}^1, \sigma_{ij}^2, ... , \sigma_{ij}^i \}$, and
\item $\pi_{ij} = \sum_{k=1}^i \pi_{ij}^k$ and
	$\pi_i^1 = \sum_{j=1}^{2^i} \pi_{ij}^1$. 
\end{itemize}
Note that:
\begin{align}
\label{eq:ProbVt}
\P_{\sigma_{01}^1}(\Future^t \in V_t) = \frac{1-p_0}{t^2} ~,
	\mbox{ for all } t \geq 1 ~,
\end{align}
and:
\begin{align}
\label{eq:p_i}
p_i = \frac{2i+1}{(i+1)^2} \leq \frac{2}{i}, \mbox{ for all } i \geq 1.
\end{align}
These facts will be used in the proof of the Claim \ref{cla:PositiveRecurrent}. 

\begin{Cla}
\label{cla:PositiveRecurrent} 
The underlying Markov chain over states for the HMM is positive recurrent. 
\end{Cla}

\begin{proof}
Let $\tau_{\cs_{01}^1} = \min \{t > 0: \CS_t = \cs_{01}^1 \}$ be the first return time to state $\cs_{01}^1$. Then, by continuity:  
\begin{align*}
\P_{\sigma_{01}^1}(\tau_{\cs_{01}^1} = \infty)	
& = \lim_{t \to \infty} \P_{\cs_{01}^1}(\tau_{\cs_{01}^1} > 2t) \\
& = \lim_{t \to \infty} \P_{\cs_{01}^1} (\Future^{t+1} \in V_{t+1}) \\
& = \lim_{t \to \infty} \frac{1-p_0}{(t+1)^2} \\
& = 0 ~.
\end{align*} 
Hence, the Markov chain is recurrent and we have:
\begin{align*}
\Ex_{\sigma_{01}^1}(\tau_{\cs_{01}^1})
& = \sum_{t =1}^{\infty} \P_{\cs_{01}^1}(\tau_{\cs_{01}^1} = t) \cdot t \\
& = p_0 \cdot 1 + \sum_{t =1}^{\infty} \P_{\cs_{01}^1}(\tau_{\cs_{01}^1} = 2t) \cdot 2t \\
& = p_0 + \sum_{t =1}^{\infty} \P_{\cs_{01}^1}(\Future^t \in V_t) \cdot p_t \cdot 2t \\
& \leq  p_0 + \sum_{t =1}^{\infty} \frac{1-p_0}{t^2} \cdot \frac{2}{t} \cdot 2t \\
& < \infty ~,
\end{align*}
from which it follows that the chain is also positive recurrent. Note that the topology 
of the chain implies the first return time may not be an odd integer greater than 1. 
\end{proof}

\begin{Cla}
\label{cla:StationaryDistribution} 
The stationary distribution $\pi$ has:
\begin{align} 
\label{eq:pi_ij1}
\pi_{ij}^1 & =  \frac{C}{i^2 \cdot 2^i}
	~~,~~ i \geq 1 ,~ 1 \leq j \leq 2^i ~, \\
\label{eq:pi_ijk}
\pi_{ij}^k & =  \frac{C}{i^2 \cdot 2^i} \cdot \frac{2i +
1}{(i+1)^2} ~~,~~ i \geq 1 , ~ 1 \leq j \leq 2^i , ~ 2 \leq k \leq i ~,
\end{align}
where $C = \pi_{01}^1 (1-p_0)$. 
\end{Cla}

\begin{proof}
Existence of a unique stationary distribution $\pi$ is guaranteed by Claim 
\ref{cla:PositiveRecurrent}. Given this, clearly
$\pi_1^1 = \pi_{01}^1 (1 - p_0)$. Similarly, for $i \geq 1$,
$\pi_{i+1}^1 = \pi_i^1 (1 - p_i) = \pi_i^1 \frac{i^2}{(i+1)^2}$, from which it
follows by induction that $\pi_i^1 =  \pi_{01}^1 (1 - p_0)/i^2$, for all
$i \geq 1$. By symmetry $\pi_{ij}^1 = \pi_i^1/2^i$ for each $i \in \N$ and
$1 \leq j \leq 2^i$. Therefore, for each
$i \in \N$, $1 \leq j \leq 2^i$ we have $\pi_{ij}^1 =  \pi_{01}^1 (1 - p_0)/(i^2 \cdot 2^i) =  C/(i^2 \cdot 2^i)$ 
as was claimed. Moreover, $\pi_{ij}^2 = \pi_{ij}^1 \cdot p_i =  \pi_{ij}^1 \cdot \frac{2i + 1}{(i+1)^2}$. Combining with
the expression for $\pi_{ij}^1$ gives $\pi_{ij}^2 =  \frac{C}{i^2 \cdot 2^i} \cdot \frac{2i + 1}{(i+1)^2}$. 
By induction, $\pi_{ij}^2 = \pi_{ij}^3 = ... ~ =  \pi_{ij}^i$, so this completes the proof. 
\end{proof}

Note that for all $i \geq 1$ and $1 \leq j \leq 2^i$:
\begin{align}
\label{eq:pi_ij_lowerbound}
\pi_{ij} & = \frac{C}{2^i \cdot i^2} + (i-1) \frac{C}{2^i \cdot i^2} \cdot \frac{2i+1}{(i+1)^2} \geq  \frac{C}{2^i \cdot i^2} ~, \mbox{ and } \\
\label{eq:pi_ij_upperbound}
\pi_{ij} & = \frac{C}{2^i \cdot i^2} + (i-1) \frac{C}{2^i \cdot i^2} \cdot \frac{2i+1}{(i+1)^2} \leq  \frac{3C}{2^i \cdot i^2} ~.
\end{align} 
Also note that for any $t \in \N$ and $i \geq 2t$ we have for each $1 \leq j \leq 2^i$:
\begin{enumerate}
\item $\P(\Future^t \in W_t|\CS_0 = \cs_{ij}^k) = 1$, for $2 \leq k \leq \lceil i/2 \rceil +1$.
\item	$\left( \sum_{k=2}^i \pi_{ij}^k \right)/\pi_{ij} \geq 1/3$ and 
         	$|\{k:2 \leq k \leq \lceil i/2 \rceil +1\}| \geq \frac{1}{2} \cdot |\{k:2 \leq k \leq i\}|$.
	Hence, $\left( \sum_{k=2}^{\lceil i/2 \rceil +1} \pi_{ij}^k \right)/\pi_{ij} \geq 1/6$. 
\end{enumerate}
Therefore, for each $t \in \N$:
\begin{align}
\label{eq:PrFuturetGood}
& \P(\Future^t \in W_t| \CS_0 \in R_{ij}) \geq 1/6 ~,
	\mbox{ for all } i \geq 2t \text{ and } 1 \leq j \leq 2^i ~.
\end{align}
Equations (\ref{eq:pi_ij_lowerbound}), (\ref{eq:pi_ij_upperbound}), and (\ref{eq:PrFuturetGood}) will be used
in the proof of Claim \ref{cla:ProbWt} below, along with the following simple lemma.
\begin{Lem}[Integral Test]
\label{lem:IntegralTest}
Let $n \in \N$ and let $f:[n,\infty] \rightarrow \R$ be a
positive, continuous, monotone-decreasing function, then:
\begin{align*}
\int_n^{\infty} f(x) dx \leq \sum_{k=n}^{\infty} f(k) \leq f(n) + \int_n^{\infty} f(x) dx ~. 
\end{align*}
\end{Lem}
\noindent

\begin{Cla}
\label{cla:ProbWt}
$\P(W_t)$ decays roughly as $1/t$. More exactly, $C/12t \leq \P(W_t) \leq 6C/t$
for all $t \in \N$. 
\end{Cla}

\begin{proof}
For any state $\sigma_{ij}^k$ with $i < t$, $\P(\Future^t \in W_t|\CS_0 = \sigma_{ij}^k) = 0$. Thus, we have:
\begin{align}
\label{eq:ProbWt}
\P(W_t) 	& = \P(\Future^t \in W_t) \nonumber \\
		& = \sum_{i = t}^{\infty} \sum_{j=1}^{2^i} \P(\CS_0 \in R_{ij}) \cdot \P(\Future^t \in W_t|\CS_0 \in R_{ij}) \nonumber \\
		& = \sum_{i = t}^{\infty}  2^i \cdot \P(\CS_0 \in R_{i1}) \cdot \P(\Future^t \in W_t|\CS_0 \in R_{i1}) ~,
\end{align}
where the second equality follows from symmetry. We prove the bounds from above 
and below on $\P(W_t)$ separately using Eq. (\ref{eq:ProbWt}).
\begin{itemize}
\item \emph{Bound from below}: 
\vspace{-5 mm}
\begin{align}
\P(W_t) 	& = \sum_{i = t}^{\infty}  2^i \cdot \P(\CS_0 \in R_{i1}) \cdot \P(\Future^t \in W_t|\CS_0 \in R_{i1}) \nonumber \\
		& \geq  \sum_{i = 2t}^{\infty}  2^i \cdot \P(\CS_0 \in R_{i1}) \cdot \P(\Future^t \in W_t|\CS_0 \in R_{i1}) \nonumber \\
		& \stackrel{(a)}{\geq} \sum_{i = 2t}^{\infty} 2^i \cdot \frac{C}{2^i \cdot i^2} \cdot \frac{1}{6} \nonumber \\
		& = \frac{C}{6}  \sum_{i = 2t}^{\infty} \frac{1}{i^2} \nonumber \\
		& \stackrel{(b)}{\geq} \frac{C}{6} \int_{2t}^{\infty} \frac{1}{x^2} dx \nonumber \\
		& = \frac{C}{12t}  ~.
\end{align}
Here, (a) follows from Eqs. (\ref{eq:pi_ij_lowerbound}) and
(\ref{eq:PrFuturetGood}) and (b) from Lemma \ref{lem:IntegralTest}.
\item \emph{Bound from above}: 
\vspace{- 5 mm}
\begin{align}
\P(W_t) 	& = \sum_{i = t}^{\infty}  2^i \cdot \P(\CS_0 \in R_{i1}) \cdot \P(\Future^t \in W_t|\CS_0 \in R_{i1}) \nonumber \\
		& \stackrel{(a)}{\leq} \sum_{i = t}^{\infty} 2^i \cdot \frac{3C}{2^i \cdot i^2} \cdot 1 \nonumber \\
		& = 3C  \sum_{i = t}^{\infty} \frac{1}{i^2} \nonumber \\
		& \stackrel{(b)}{\leq} 3C \left( \frac{1}{t^2} + \int_{t}^{\infty} \frac{1}{x^2} dx  \right) \nonumber \\
		& = 3C \cdot \left( \frac{1}{t^2} + \frac{1}{t} \right) \nonumber \\
		& \leq \frac{6C}{t} ~.
\end{align}
Here, (a) follows from Eq. (\ref{eq:pi_ij_upperbound}) and (b) from Lemma \ref{lem:IntegralTest}.
\end{itemize}
\end{proof}

\begin{Cla}
\label{cla:ProbWt_givenw}
$\P(\MS_t \in \{2,3\} | \Future^t = w) \geq 1/150$, for all $t \in \N$ and
$w \in W_t$. 
\end{Cla}

\begin{proof}
Applying Claim \ref{cla:ProbWt} we have for any $t \in \N$: 
\begin{align*}
\P(\MS_t \in \{2,3\} | \Future^t \in W_t) 
& = \P(\Future^{t+1} \in W_{t+1} | \Future^t \in W_t) \\
& = \P(\Future^{t+1} \in W_{t+1},  \Future^t \in W_t) / \P(\Future^t \in W_t) \\
& = \P(\Future^{t+1} \in W_{t+1}) / \P(\Future^t \in W_t) \\
& \geq \frac{ C/12(t+1) }{ 6C/t} \\
& = \frac{1}{72} \cdot \frac{t}{t+1} \\
& \geq \frac{1}{150} ~.
\end{align*}
By symmetry, $\P(\MS_t \in \{2,3\} | \Future^t = w)$ is the same for each $w \in W_t$. Thus, the same bound
must also hold for each $w \in W_t$ individually: $\P(\MS_t \in \{2,3\} | \Future^t = w) \geq 1/150$ for all $w \in W_t$.
\end{proof}

\begin{Cla} 
\label{cla:hw_and_thw}
For each $t \in \N$ and $w \in W_t$, (i) $\th_w \leq 1/300$ and (ii) $h_w \geq 1/150$. Hence, $h_w - \th_w \geq 1/300$.
\end{Cla}

\begin{proof}[Proof of (i)] 
$h_{\sigma_{ij}^k} = 0$, for all $i \geq 1$, $1 \leq j \leq 2^i$, and
$k \geq 2$. And, for each $w \in W_t$, $\phi(w)_{\cs_{ij}^1} = 0$, for all
$i \geq 1$ and $1 \leq j \leq 2^i$. 
Hence, for each $w \in W_t$, $\th_w =  \sum_{\sigma \in \CSSet} \phi(w)_{\sigma} h_{\sigma} = \phi(w)_{ \sigma_{01}^1}  h_{ \sigma_{01}^1}$.
By construction of the machine $h_{ \sigma_{01}^1} \leq 1/300$ and,
clearly, $\phi(w)_{ \sigma_{01}^1}$ can never exceed $1$.
Thus, $\th_w \leq 1/300$ for all $w \in W_t$.
\end{proof}

\begin{proof}[Proof of (ii)] 
Let the random variable $Z_t$ be defined by: $Z_t = 1$ if $\MS_t \in \{2,3\}$
and $Z_t = 0$ if $\MS_t \not\in\{2,3\}$. By Claim \ref{cla:ProbWt_givenw},
$\P(Z_t = 1|\Future^t = w) \geq 1/150$ for any $w \in W_t$ and, by symmetry,
the probabilities of a $2$ or a $3$ following any word $w \in W_t$ are equal,
so $\P(\MS_t = 2|\Future^t = w,Z_t=1) = \P(\MS_t = 3|\Future^t = w,Z_t=1) = 1/2$. 
Therefore, for any $w \in W_t$:
\begin{align*}
h_w 	& = H[\MS_t|\Future^t=w] \\
	& \geq H[\MS_t|\Future^t=w,Z_t] \\
	& \geq  \P(Z_t = 1|\Future^t = w) \cdot H[\MS_t|\Future^t=w, Z_t=1] \\
	& \geq 1/150 \cdot 1 ~.
\end{align*}
\end{proof}

\begin{Cla}
\label{cla:hmut_decay}
The quantity $\hmu(t) - \hmu$ decays at a rate no faster than $1/t$. More exactly, $\hmu(t+1) - \hmu \geq \frac{C}{3600 t}$,  for all $t \in \N$. 
\end{Cla}

\begin{proof}
As noted above, since the machine satisfies the conditions of Prop.
\ref{prop:ExactEntropyRateFormula}, the entropy rate is given by
Eq. (\ref{eq:UnifilarEntropyRateFormula}) and the difference 
$\hmu(t+1) - \hmu$ is given by Eq. (\ref{eq:hmut_hmu_diff}).
Therefore, applying Claims \ref{cla:ProbWt} and \ref{cla:hw_and_thw} we
may bound the quantity $\hmu(t+1) - \hmu$ as follows:
\begin{align*}
\hmu(t+1) - \hmu 	& = \sum_{w \in \L_t} \P(w) (h_w - \th_w) \\
				& \geq \sum_{w \in W_t}  \P(w) (h_w - \th_w) \\
				& \geq \P(W_t) \cdot \frac{1}{300} \\
				& \geq \frac{C}{3600 t} ~.
\end{align*}
\end{proof}
With the above decay on $\hmu(t)$ established we easily see the Branching Copy Process
must have infinite excess entropy. 

\begin{Prop}
The excess entropy $\EE$ for the BC Process is infinite.
\end{Prop}

\begin{proof}
$\EE = \sum_{t=1}^{\infty} \left( \hmu(t) - \hmu \right)$. By Claim
\ref{cla:hmut_decay}, this sum must diverge.  
\end{proof}

\section{Conclusion}
\label{sec:Conclusion}

Any stationary, finite-alphabet process may be represented as an invariant HMM
with an uncountable state set. Thus, there exist invariant HMMs with
uncountable state sets capable of generating infinitary processes over finite
alphabets. It is impossible, however, to have a finite-state invariant HMM that 
generates an infinitary process. The excess entropy $\EE$ is always
bounded by the entropy in the stationary distribution $H[\pi]$, which is finite
for any finite-state HMM. Countable-state HMMs are intermediate between the
finite and uncountable cases, and it was previously unknown whether infinite
excess entropy was possible in this case. We have demonstrated that it is
indeed possible, by giving two explicit constructions of finite-alphabet
infinitary processes generated by invariant HMMs with countable state sets. 

The second example, the Branching Copy Process, is also ergodic---a strong
restriction. It is a priori quite plausible that infinite $\EE$ might only
occur in the countable-state case for nonergodic processes. Moreover, both
HMMs we constructed are unifilar, so the \eMs\ \cite{Crut88a, Lohr10a} of
the processes have countable state sets as well. Again, unifilarity is a
strong restriction to impose, and it is a priori conceivable that infinite
$\EE$ might only occur in the countable-state case for nonunifilar HMMs.
Our examples have shown, though, that infinite $\EE$ is possible for
countable-state HMMs, even if one requires both ergodicity and unifilarity. 

\appendix

\section{}
\label{AppendixA}
We prove Prop. \ref{prop:ExactEntropyRateFormula} from Sec.
\ref{subsec:HiddenMarkovModels}, which states that the entropy rate of any
countable-state, exact, unifilar HMM is given by the standard formula:
\begin{align}
\hmu = H[\MS_0|\CS_0] = \sum_{\cs \in \CSSet} \pi_{\cs} h_{\cs} ~.
\end{align}

\begin{proof}
Let $\L_t = \{w: |w| = t, \P(w) > 0\}$ be the set of length $t$ words in the
process language $\L$, and let $\phi(w)$ be the conditional state distribution
induced by a word $w \in \L$: i.e.,
$\phi(w)_{\cs} = \P(\CS_t = \cs|\Future^t = w$). Furthermore, let
$\th_w = \sum_{\cs} \phi(w)_{\cs} h_{\cs}$ be the $\phi(w)$-weighted average
entropy in the next symbol given knowledge of the current state $\cs$. And let 
$h_w = H[\MS_t|\Future^t = w] = H[\MS_0|\CS_0 \sim \phi(w)]$ be the entropy in the next symbol given the word $w$. Note that:
\begin{enumerate}
\item $\hmu(t+1) = H[\MS_t|\Future^t] = \sum_{w \in \L_t} \P(w) h_w$ , and
\item $ \sum_{\cs} \pi_{\cs} h_{\cs}
	= \sum_{\cs} \left( \sum_{w \in \L_t} \P(w) \phi(w)_{\cs} \right) h_{\cs}
	= \sum_{w \in \L_t} \P(w) \left( \sum_{\cs} \phi(w)_{\cs} h_{\cs} \right)
	= \sum_{w \in \L_t} \P(w) \th_w $ .
\end{enumerate}
Since we know $\hmu(t)$ limits to $\hmu$, it suffices to show that:
\begin{align}
\label{eq:zero_diff}
\lim_{t \to \infty} \sum_{w \in \L_t} \P(w) (h_w - \th_w) = 0 ~.
\end{align} 
By concavity of the entropy function, $h_w - \th_w \geq 0$ for any $w$. 
However, for a \emph{synchronizing word} $w = w_1 ... w_t$ with $H[\CS_t|\Future^t = w] = 0$, $h_w - \th_w$  
is always 0, since the distribution $\phi(w)$ is concentrated only
on a single state. Furthermore, for any $w$, 
$h_w - \th_w \leq h_w \leq \log|\XX|$. Thus:
\begin{align}
\label{eq:hmut_hmu_diff_bound}
\sum_{w \in \L_t} \P(w) (h_w - \th_w) \leq \log|\XX| \cdot \P(NS_t) ~ ,
\end{align} 
where $NS_t$ is the set of length-$t$ words that are nonsynchronizing and
$\P(NS_t)$ is the combined probability of all words in this set. Since the
HMM is exact, we know that for a.e. infinite future $\future$ an observer
will synchronize exactly at some finite time $t = t(\future)$.  And, since
it is unifilar, the observer will remain synchronized for all $t' \geq t$.
It follows that $\P(NS_t)$ must be monotonically decreasing
and limit to $0$:
\begin{align}
\label{eq:PNStLimit}
\lim_{t \to \infty} \P(NS_t) = 0 ~.
\end{align} 
Combining Eq. (\ref{eq:hmut_hmu_diff_bound}) with Eq. (\ref{eq:PNStLimit})
shows that Eq. (\ref{eq:zero_diff}) does in fact hold,
which completes the proof. 
\end{proof}  

\section{}
\label{AppendixB}

We prove the following proposition for the entropy rate of
countable-state HMMs.

\begin{Prop}
Let $M$ be a countable-state HMM and let $\Process = (X_t)$ be the process
generated by $M$. If $\Process$ does not consist entirely of periodic
sequences, then its entropy rate $\hmu$ $\Process$ is strictly positive. 
\end{Prop} 

\begin{proof}
For any countable-state HMM $M$, the future output sequence and past output
sequence are conditionally independent given the current state. Thus, for all $t \in \N$, 
$H[\MS_t|\Future^t, \CS_t] = H[\MS_t|\CS_t]$. Also, by stationarity 
$H[\MS_t|\CS_t] = H[\MS_0|\CS_0] = \sum_{\cs} \pi_{\cs} h_{\cs}$, for all $t$.
Combining these facts shows that entropy rate is always bounded below by the 
standard unifilar formula of Eq. (\ref{eq:UnifilarEntropyRateFormula}):
\begin{align}
\label{eq:hmu_lowerbound}
\hmu 	& = \lim_{t \to \infty} H[\MS_t|\Future^t] \nonumber \\
		& \geq \lim_{t \to \infty} H[\MS_t|\CS_t, \Future^t] \nonumber \\
		& = \lim_{t \to \infty} H[\MS_t|\CS_t] \nonumber \\
		& = \sum_{\cs \in \CSSet} \pi_{\cs} h_{\cs}  ~.
\end{align}
Therefore, the entropy rate is positive if $h_{\cs} > 0$ for any state $\cs$
with nonzero probability $\pi_{\cs}$ or, equivalently, if there are at least
two outgoing edges in the associated graph from state $\cs$. 

Now, assume there is no such state. Consider the restricted state set
$\widetilde{\CSSet}$ consisting of states $\cs$ with positive probability
($\pi_{\cs} > 0$) and the restricted graph $\widetilde{G}$ associated to
this state set. Clearly, the HMM $\widetilde{M}$ defined by this graph with stationary distribution $\pi$ generates the same process $\Process$ as the
original HMM. And, it is also easily seen that in order to keep the
distribution $\pi$ stationary, the graph $\widetilde{G}$ must consist entirely
of disjoint strongly connected components. That is, each connected component
of $\widetilde{G}$ must be strongly connected. Take any strongly connected
component $C_i$ in $\widetilde{G}$. Since each state $\cs$ in $C_i$ has only
a single outgoing edge and $C_i$ is strongly connected, it follows that $C_i$
must be a deterministic loop of some finite length $l_i$. Since this holds for
each strongly connected component $C_i$ in $\widetilde{G}$ and the HMM
$\widetilde{M}$ is always run from one of the $C_i$s, it follows that all
sequences $\biinfinity = ... \ms_{-1} \ms_{0} \ms_{1} ...$ generated by
$\widetilde{M}$ are periodic. Or, equivalently, all sequences generated by
$M$ are periodic.  
\end{proof} 

\bibliography{ref,chaos}

\end{document}

%% file: cmechabbrev.tex
\newcommand{\eM}     {$\epsilon$\protect\nobreakdash-machine}
\newcommand{\eMs}    {$\epsilon$\protect\nobreakdash-machines}

\newcommand{\Process}{\mathcal{P}}

\newcommand{\MeasSymbol}   { {X} }
\newcommand{\meassymbol}   { {x} }

\newcommand{\biinfinity}	{ \overleftrightarrow {\meassymbol} }

\newcommand{\Past}	{ \overleftarrow {\MeasSymbol} }
\newcommand{\past}	{ {\overleftarrow {\meassymbol}} }

\newcommand{\Future}	{ \overrightarrow{\MeasSymbol} }
\newcommand{\future}	{ \overrightarrow{\meassymbol} }

\newcommand{\causalstate}	{ \sigma }

\newcommand{\hmu}		{h_\mu}
\newcommand{\EE}		{{\bf E}}



%% file: shortcuts.tex
\theoremstyle{plain}   
\theoremstyle{plain}   \newtheorem{Lem}{Lemma}
\theoremstyle{plain} 	
\theoremstyle{plain} 	
\theoremstyle{plain} 	\newtheorem{Prop}{Proposition}
\theoremstyle{plain} 	
\theoremstyle{plain}	\newtheorem*{Rem}{Remark}
\theoremstyle{plain}	\newtheorem{Def}{Definition} 
\theoremstyle{plain}	
\theoremstyle{plain}   \newtheorem{Cla}{Claim}
\theoremstyle{plain} 

\def\norm#1{\|#1\|}

\def\N{\mathbb{N}}

\def\R{\mathbb{R}}

\def\Z{\mathbb{Z}}

\def\P{\mathbb{P}}

\def\Ex{\mathbb{E}}

\def\XX{\mathcal{X}}

\def\th{\tilde{h}}

\def\CS{S}
\def\cs{\sigma}

\def\MS{X}
\def\ms{x}

\def\CSSet{\mathcal{S}}

\def\L{\mathcal{L}}

